\documentclass[12pt]{article}
\usepackage{amsmath, amsthm, amsfonts, amssymb, bbm}
\usepackage{graphics}
\usepackage[english]{babel}

\usepackage{ifpdf} 
\ifpdf
\usepackage{epstopdf}
\usepackage{hyperref}
\else
\usepackage[hypertex]{hyperref}
\fi

\newtheorem{thm}{Theorem}
\newtheorem{lem}[thm]{Lemma}
\newtheorem{prop}[thm]{Proposition}
\newtheorem{rem}[thm]{Remark}
\newtheorem{defn}[thm]{Definition}

 \numberwithin{equation}{section}
 \numberwithin{thm}{section}

\DeclareMathOperator{\End}{End}
\DeclareMathOperator{\Hom}{Hom}
\DeclareMathOperator{\Dim}{Dim}
\DeclareMathOperator{\Pic}{Pic}
\DeclareMathOperator{\Rep}{Rep}

\newcommand\void[1]       {}

\newcommand\be            {\begin{equation}}
\newcommand\ee            {\end{equation}}
\newcommand\bea           {\begin{eqnarray}}
\newcommand{\eea}         {\end{eqnarray}}
\newcommand\nn             {\nonumber \\}

\newcommand\etb           {&\!\! \displaystyle}

\newcommand\labl[1]       {\label{#1}\ee}
\newcommand\nxt{\noindent\raisebox{.08em}{\rule{.44em}{.44em}}%
\hspace{.4em}}

\newcommand\arxiv[2]      {\href{http://arXiv.org/abs/#1}{#2}}
\newcommand\doi[2]        {\href{http://dx.doi.org/#1}{#2}}
\newcommand\httpurl[2]    {\href{http://#1}{#2}}

\newcommand\CxC           {\Cc\,{\boxtimes}\,\overline\Cc}
\newcommand\eps           {\varepsilon}
\newcommand\id            {{\rm id}}

\newcommand\one           {{\bf1}}

\newcommand\Zc            {\mathcal{Z}}
\newcommand\Cb            {\mathbb{C}}

\newcommand\Rb            {\mathbb{R}}

\newcommand\Cc            {\mathcal{C}}
\newcommand\Hc            {\mathcal{H}}
\newcommand\Ic            {\mathcal{I}}
\newcommand\Jc            {\mathcal{J}}

\begin{document}

\thispagestyle{empty}
\def\thefootnote{\fnsymbol{footnote}}
\begin{flushright}
ZMP-HH/10-13\\
Hamburger Beitr\"age zur Mathematik 374
\end{flushright}
\vskip 3em
\begin{center}\LARGE
Invertible defects and isomorphisms of rational CFTs
\end{center}

\vskip 2em
\begin{center}
{\large 
Alexei Davydov$^{a}$,\, Liang Kong$^{b,c}$,\,  Ingo Runkel$^{d}$,\, ~\footnote{Emails: 
{\tt davydov@mpim-bonn.mpg.de, kong.fan.liang@gmail.com, ingo.runkel@uni-hamburg.de}}}
\\[1em]
\it$^a$ Max Planck Institut f\"ur Mathematik\\
Vivatsgasse 7, 53111 Bonn, Germany
\\[1em]
$^b$ California Institute of Technology, \\
Center for the Physics of Information, \\
Pasadena, CA 91125, USA
\\[1em]
$^c$ Institute for Advanced Study (Science Hall) \\ 
Tsinghua University, Beijing 100084, China
\\[1em]
$^d$ Department Mathematik, Universit\"at Hamburg\\
  Bundesstra\ss e 55, 20146 Hamburg, Germany
\end{center}

\vskip 2em
\begin{center}
  April 2010
\end{center}
\vskip 2em

\begin{abstract}
Given two two-dimensional conformal field theories, a domain wall -- or defect line -- between them is called invertible if there is another defect with which it fuses to the identity defect. A defect is called topological if it is transparent to the stress tensor. A conformal isomorphism between the two CFTs is a linear isomorphism between their state spaces which preserves the stress tensor and is compatible with the operator product expansion. We show that for rational CFTs there is a one-to-one correspondence between invertible topological defects and conformal isomorphisms if both preserve the rational symmetry. This correspondence is compatible with composition.
\end{abstract}

\setcounter{footnote}{0}
\def\thefootnote{\arabic{footnote}}

\newpage

\tableofcontents

\section{Introduction}

Dualities play an important role in understanding non-perturbative properties of models in quantum field theory, statistical physics or string theory, because they allow to relate observables in a model at weak coupling to those of the dual model at strong coupling. Some well known examples are 
Kramers-Wannier duality \cite{Kramers:1941kn}, 
electric-magnetic duality \cite{Montonen:1977sn}, 
T-duality \cite{Giveon:1994fu},
mirror symmetry \cite{Lerche:1989uy,Greene:1990ud}, 
and 
the AdS/CFT correspondence \cite{Maldacena:1997re}.

By their very nature, dualities are hard to find and it is difficult to understand precisely how quantities in the two dual descriptions are related. In many examples, it has proved helpful to describe dualities by a `duality domain wall', a co-dimension one defect which separates the dual theories  \cite{Frohlich:2004ef,Fuchs:2007fw,Fuchs:2007tx,Gaiotto:2008ak,Kapustin:2009av}. It is then natural to ask if in any sense {\em all} dualities can be described by such defects.
For a particularly simple type of duality defects -- so-called invertible defects -- in a particularly well understood class of quantum field theories, namely two-dimensional rational conformal field theories, we will answer this question in the affirmative. Let us describe the setting and the result of this paper in more detail.

\medskip

Generically, a duality transformation exchanges local fields and disorder fields. This is the case in the archetypical example of such dualities, Kramers-Wannier duality of the two-dimensional Ising model. In the lattice model, the duality exchanges the local spin-operator with the non-local disorder-operator, which marks the endpoint of a frustration line on the dual lattice. In the conformal field theory which describes the critical point of the Ising model, the duality accordingly provides an automorphism on the space consisting of all local fields and all disorder fields. In particular, the Kramers-Wannier duality transformation is not an automorphism on the space of local fields alone. 

However, there is an especially simple type of duality which does give rise to an isomorphism between the spaces of local fields for the two models related by the duality. The conformal field theory description of T-duality and mirror symmetry on the string world sheet are examples of such dualities. Given two conformal field theories $C^A$ and $C^B$, the data of such a duality consists of an isomorphism between their spaces of states $\Hc^A$ and $\Hc^B$ which respects the operator product expansion and which preserves the vacuum and the stress tensor; we will call this a {\em conformal isomorphism}.

The infinite symmetry algebra of a conformal field theory is generated by its conserved currents. It always includes the stress tensor, accounting for the Virasoro symmetry, but it may also contain fields that do not arise via multiple operator product expansions of the stress tensor. A {\em rational} CFT, roughly speaking, is a CFT whose symmetry algebra is large enough to decompose the space of states into a {\em finite} direct sum of irreducible representations. Examples of rational CFTs are the Virasoro minimal models, rational toroidal compactifications of free bosons, Wess-Zumino-Witten models and coset models obtained from affine Lie algebras at positive integer level, as well as appropriate orbifolds thereof. 

Suppose that we are given two CFTs $C^A$ and $C^B$ which are rational, have a unique vacuum, have isomorphic algebras of holomorphic and anti-holomorphic conserved currents, and have a modular invariant partition function. We will show that for each conformal isomorphism that preserves the rational symmetry, there exists (up to isomorphism) one and only one invertible defect, i.e.\ a duality domain wall, between the CFTs $C^A$ and $C^B$ which implements this duality. Conversely, each invertible defect gives rise to a conformal isomorphism. Altogether we show that for this class of models:
\be
\begin{minipage}{33em}
\em
There is a bijection between conformal isomorphisms and invertible defects,\\ both preserving
the rational symmetry.
\end{minipage}
\labl{eq:phys_version}
The proof relies on the vertex algebraic description of CFTs in \cite{Huang:2005hk,Kong:2006}, on the relation between two-dimensional CFT and three-dimensional topological field theory \cite{Felder:1999cv,Fuchs:2001am,Fjelstad:2006aw,Frohlich:2006ch}, and on results in categorial algebra \cite{Kong:2007yv,Kong:2008ci}. 
Given this background, the proof is actually quite short, and it is phrased as a result in categorical algebra. 
Let us briefly link the physical concepts with their mathematical counterparts; more details and the proof will be given in Section~\ref{sec:pf}. 

The representations of the holomorphic chiral algebra of a rational CFT (a vertex operator algebra) form a so-called modular category \cite{Moore:1988qv,Turaev:1994,Huang:2005hh}, which we denote by $\Cc$. The bulk fields of a rational CFT with unique vacuum and with isomorphic holomorphic and anti-holomorphic chiral algebra give rise to a simple commutative symmetric Frobenius algebra $C$ in $\CxC$ \cite{Kong:2006,Fjelstad:2006aw}. Here, $\CxC$ is the product of two copies of $\Cc$, where the second copy corresponds to representations of the anti-chiral algebra (so that the braiding and twist there are replaced by their inverses). We assume in addition that the CFT is modular invariant. In this case the algebra $C$ is maximal, a condition on the categorical dimension of $C$ defined in Section~\ref{sec:pf}.
If the CFT is defined on the upper half plane and the boundary condition preserves the rational chiral symmetry, the boundary fields give rise to a simple special symmetric Frobenius algebra $A$ in the modular category $\Cc$ \cite{tft1,cardy-cond}. From $A$ one can construct the {\em full centre} $Z(A)$, a simple commutative maximal special symmetric Frobenius algebra in $\CxC$ \cite{Fjelstad:2006aw}. It is proved in \cite{Fjelstad:2006aw,Kong:2008ci} that $C \cong Z(A)$ as algebras. Denote by $\Cc_{A|A}$ the monoidal category of $A$-$A$-bimodules in $\Cc$. These bimodules describe topological defect lines of the CFT which preserve the chiral symmetry \cite{tft1,Frohlich:2006ch}. Invertible topological defects correspond to invertible $A$-$A$-bimodules. 
Let $\Pic(A)$ be the Picard group of $\Cc_{A|A}$. The elements of $\Pic(A)$ are isomorphism classes of invertible objects in $\Cc_{A|A}$ and the group operation is induced by the tensor product of $\Cc_{A|A}$. 
We prove that there is an isomorphism of groups
\be
  \mathrm{Aut}(Z(A)) \cong \Pic(A) \ ,
\ee
where $\mathrm{Aut}(Z(A))$ are the algebra automorphisms of $Z(A)$. In fact, we will prove a groupoid version of this statement. The first groupoid has as objects simple special symmetric Frobenius algebras in $\Cc$ and as morphisms isomorphism classes of invertible bimodules. The second groupoid has simple commutative maximal special symmetric Frobenius algebras in $\CxC$ as objects and its morphisms are algebra isomorphisms. We prove the equivalence of these two groupoids, which is the mathematical version of the physical statement \eqref{eq:phys_version}.

\bigskip

This paper is organised as follows. In Section~\ref{sec:d+d} we give a brief description of CFT and defect lines, and we formulate the result of the paper this language. In Section~\ref{sec:pf}, the result is restated in algebraic terms and proved. Section~\ref{sec:ex} contains two examples, and with Section~\ref{sec:conc} we conclude.

\section{Conformal isomorphisms and defects}\label{sec:d+d}

Consider a CFT $C^A$ with space of states $\Hc^A$. By the state-field correspondence, $\Hc^A$ coincides with the space of fields of the CFT. The space of states contains the states $T^A$ and $\bar T^A$, the holomorphic and anti-holomorphic components of the stress tensor. Their modes,  $L_m$ and $\bar L_m$, give rise to two commuting copies of the Virasoro algebra. Pick a basis $\{ \phi_i \}$ of $\Hc^A$ consisting of eigenvectors\footnote{
  We assume here that $L_0$ and $\bar L_0$ are diagonalisable, i.e.\ we exclude 
  logarithmic CFTs from our treatment. We also assume the common eigenspaces of $L_0$ and
  $\bar L_0$ are finite-dimensional, and that their eigenvalues form a countable set. The latter condition
  excludes for example Liouville theory.}
of $L_0$ and $\bar L_0$. Then we have the operator product expansion (OPE)
\be
  \phi_i(z) \phi_j(w) = \sum_k C^A_{ijk}(z{-}w) \, \phi_k(w) \ ,
\ee
where $z$ and $w$ are two distinct points on the complex plane and each function $C^A_{ijk}(x)$ is determined by conformal covariance up to an overall constant; the OPE has to be associative and commutative \cite{Belavin:1984vu}, see \cite{Huang:2005hk} for the mathematical formulation we will use in Section~\ref{sec:pf}. Apart from an associative commutative OPE, we make the following assumptions:

\medskip\noindent
{\em Uniqueness of the vacuum:} There is a unique element $\one^A \in \Hc^A$, the vacuum vector, which is annihilated by $L_0$, $L_{\pm1}$ and $\bar L_0$, $\bar L_{\pm1}$, and which has the OPE $\one(z)\one(w) = \one(w)$.

\medskip\noindent
{\em Non-degeneracy:} Take the first basis vector to be $\phi_1 = \one^A$. Then $\langle \phi_i, \phi_j \rangle := C^A_{ij1}$ defines a non-degenerate pairing on the space of states $\Hc^A$. In other words, the two-point correlator is non-degenerate.

\medskip\noindent
{\em Modular invariance:} The partition function $Z(\tau) = \text{tr}_{\Hc^A}\, q{}^{L_0-c/24} (q^*){}^{\bar L_0 - \bar c/24}$ is modular invariant, i.e.\ it obeys $Z(\tau) = Z(-1/\tau) = Z(\tau+1)$. Here $\tau$ is a complex number with $\text{im}(\tau)>0$, $q = \exp(2 \pi i \tau)$, and $c$ and $\bar c$ are the  left and the right central charge.

\medskip

Suppose now we are given two CFTs $C^A$ and $C^B$. By a {\em conformal isomorphism} $f$ from $C^A$ to $C^B$ we mean a linear isomorphism $f : \Hc^A \rightarrow \Hc^B$ which preserves the vacuum, the stress tensor, and the OPE. This means that $f(\one^A) = \one^B$, $f(T^A)=T^B$, $f(\bar T^A) = \bar T^B$, and that, if we choose a basis $\{ \phi_i \}$ of $\Hc^A$ as above, and take $\phi_i'=f(\phi_i)$ as basis for $\Hc^B$, then $C^A_{ijk}(x) = C^B_{ijk}(x)$.

\bigskip

Next we give some background on defects. Given two CFTs $C^A$ and $C^B$, we can consider domain walls -- or defects -- between $C^A$ and $C^B$. To be specific, take the complex plane with a defect placed on the real axis, and with CFT $C^A$ defined on the upper half plane and CFT $C^B$ on the lower half plane. The defect is defined by the boundary conditions obeyed by the fields of $C^A$ and $C^B$ on the real line. We call a defect {\em conformal} iff the stress tensors satisfy $T^A(x)-\bar T^A(x) = T^B(x)-\bar T^B(x)$ for all $x\in\Rb$. The defect is called {\em topological} iff the stronger conditions $T^A(x)=T^B(x)$ and $\bar T^A(x)=\bar T^B(x)$ hold for all $x\in\Rb$. Topological defects are totally transmitting and tensionless. They can exist only if the central charges of the CFTs $C^A$ and $C^B$ are the same, and they can be deformed on the complex plane without affecting the values of correlators, as long as they are not taken past field insertions or other defects. A trivial example of a topological defect is the {\em identity defect} between a given CFT and itself, which simply consists of no defect at all, i.e.\ all fields of the CFT are continuous across the real line.

Conformal defects are very difficult to classify, the only models for which all conformal defects (with discrete spectrum) are known are the Lee-Yang model and the critical Ising model \cite{Oshikawa:1996dj,Quella:2006de}; even for the free boson one knows only certain examples \cite{Bachas:2001vj,Bachas:2007td}. Topological defects have been classified for Virasoro minimal models \cite{Petkova:2000ip,tft1} and for the free boson \cite{Fuchs:2007tx}.

For topological defects one can define the operation of fusion \cite{Petkova:2000ip,tft1}, whereby one places a topological defect $R$ on the real line, and another topological defect $S$ on the line $\Rb + i\eps$, and considers the limit $\eps \rightarrow 0$. Since correlators are independent of $\eps$, this procedure is non-singular (which is not true for general conformal defects \cite{Bachas:2001vj,Bachas:2007td}), and it gives a new topological defect $R \star S$ on the real line. We call a topological defect between CFTs $C^A$ and $C^B$ {\em invertible}, iff 
there exists a defect between $C^B$ and $C^A$ such that their fusion in both possible orders yields the identity defect of CFT $C^A$ and of CFT $C^B$, respectively.

A topological defect $R$ between CFT $C^A$ and CFT $C^B$ gives rise to a linear operator $D[R] : \Hc^A \rightarrow \Hc^B$. This operator is obtained by placing a field $\phi$ of CFT $C^A$ at the origin $0$ and the defect $R$ on the circle around 0 of radius $\eps$. In the limit $\eps \rightarrow 0$ (again, all correlators are actually independent of $\eps$) one obtains a field $\psi$ of CFT $C^B$. This defines the action of $D[R]$ via $\psi = D[R]\phi$. Since the defect is topological, $D[R]$ intertwines the Virasoro actions on $\Hc^A$ and $\Hc^B$. The identity defect induces the identity map, and the assignment is compatible with fusion of defects, $D[R \star S] = D[R]D[S]$. In particular, invertible defects give rise to isomorphisms between state spaces.

Given two (non-trivial) CFTs $C^A$ and $C^B$, it is not true that every linear map from $\Hc^A$ to $\Hc^B$ can be written as $D[R]$ for an appropriate defect $R$. Indeed, a defect has to satisfy many additional conditions. One way to formulate this is to extend the axiomatic definition of CFT in terms of sewing of surfaces \cite{Segal} to surfaces decorated by defect lines \cite{Runkel:2008gr}. For example, in the setting of \cite{Runkel:2008gr}, one can show that an invertible defect $X$ between $C^A$ and $C^B$ provides a conformal isomorphism $Z(X)$ from $C^A$ to $C^B$ by setting $Z(X) = \gamma_X^{\,-1} D[X]$, where $\gamma_X \in \Cb$ is defined via $D[X] \one^A = \gamma_X \one^B$.

\bigskip

Let us now restrict our attention to rational CFTs. 
More precisely, by a rational CFT $C^A$ we mean that $\Hc^A$ contains a subspace $V_L$ consisting of holomorphic fields and $\bar V_R$ of anti-holomorphic fields, such that $V_L$ and $V_R$ are vertex operator algebras (VOAs) satisfying the conditions of \cite{Huang:2005hh}, and such that $V_L \otimes_\Cb \bar V_R$ is embedded in $\Hc^A$ (the bar in $\bar{V}_R$ just reminds us that the fields in $V_R$ are anti-holomorphic). This turns $\Hc^A$ into a $V_L \otimes_\Cb \bar V_R$-module, and by rationality of $V_L \otimes_\Cb \bar V_R$ it is finitely reducible, see \cite{Huang:2005hk} for details. We call $C^A$ a rational CFT over $V_L \otimes_\Cb \bar V_R$. Note that, while bulk fields in the image of $V_L \otimes_\Cb \bar V_R$ can always be written as a sum of (non-singular) OPEs of a holomorphic and an anti-holomorphic field in $\Hc^A$, the same is in general not true for an arbitrary field in $\Hc_A$.

Given two rational CFTs $C^A$ and $C^B$ over $V_L \otimes_\Cb \bar V_R$, 
we say a conformal isomorphism from $\Hc^A$ to $\Hc^B$ {\em preserves the rational symmetry}
iff it acts as the identity on $V_L \otimes_\Cb \bar V_R$. Similarly we say that a defect from $C^A$ to $C^B$ {\em preserves the rational symmetry}
iff all bulk fields in $V_L \otimes_\Cb \bar V_R$ are continuous across the defect line. Since $T$ and $\bar T$ are in $V_L \otimes_\Cb \bar V_R$, such a defect is in particular topological. 

\medskip

We have now gathered in more detail all the ingredients needed to state our main result: Given two rational CFTs $C^A$ and $C^B$ over $V \otimes_\Cb \bar V$ (i.e.\ we demand that $V_L=V_R=V$), for each conformal isomorphism $f$ from $C^A$ to $C^B$ there exists a unique (up to isomorphism, see Section~\ref{sec:pf}) invertible defect $X$ such that $f = Z(X)$. This assignment is compatible with composition.

As a special case of this result we obtain that all automorphisms of a rational CFT over $V \otimes_\Cb \bar V$ which act as the identity on $V \otimes_\Cb \bar V$ are implemented by defects. The existing results in the literature \cite{Fuchs:2007vk} imply that there is an injective group homomorphism from (isomorphism classes of) invertible defects of the CFT to itself to conformal automorphisms. Our result shows in addition that this map is surjective.
Let us stress that this is by no means obvious, as the defining conditions to be satisfied by conformal isomorphisms and defects are very different: compatibility with the OPE versus sewing relations for surfaces decorated by defect lines.

\section{Proof via algebras in modular categories}\label{sec:pf}

The aim of this section is to prove an equivalence of groupoids which is the algebraic counterpart of the CFT result stated in \eqref{eq:phys_version} and detailed in 
the previous section. We will start by introducing the necessary algebraic objects -- modular categories, certain Frobenius algebras, the full centre -- and describe their relation to CFT in a series of remarks.

\subsection{Modular categories}

We will employ the usual graphical notation for ribbon categories \cite{js},\cite{Turaev:1994,baki}. To fix conventions, we note that our diagrams are read from bottom to top (the `optimistic' way), and that the pictures for the braiding and the duality morphisms are
\be
   \raisebox{-20pt}{
  \begin{picture}(26,45)
   \put(0,6){\scalebox{0.75}{\includegraphics{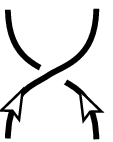}}}
   \put(0,6){
     \setlength{\unitlength}{.75pt}\put(-146,-155){
     \put(143,145)  {\scriptsize $ U $}
     \put(169,145)  {\scriptsize $ V $}
      \put(143,198)  {\scriptsize $ V $}
     \put(169,198)  {\scriptsize $ U $}
     }\setlength{\unitlength}{1pt}}
  \end{picture}}  
: U \otimes V \xrightarrow{c_{U,V}} V \otimes U
\ee  
and
\be
\begin{array}{llll}
  \raisebox{-8pt}{
  \begin{picture}(26,22)
   \put(0,6){\scalebox{.75}{\includegraphics{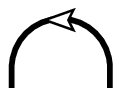}}}
   \put(0,6){
     \setlength{\unitlength}{.75pt}\put(-146,-155){
     \put(143,145)  {\scriptsize $ U^\vee $}
     \put(169,145)  {\scriptsize $ U $}
     }\setlength{\unitlength}{1pt}}
  \end{picture}}  
  \etb:U^\vee \otimes U \xrightarrow{d_U} \one,\qquad &
  \raisebox{-8pt}{
  \begin{picture}(26,22)
   \put(0,6){\scalebox{.75}{\includegraphics{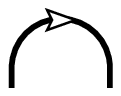}}}
   \put(0,6){
     \setlength{\unitlength}{.75pt}\put(-146,-155){
     \put(143,145)  {\scriptsize $ U $}
     \put(169,145)  {\scriptsize $ U^\vee $}
     }\setlength{\unitlength}{1pt}}
  \end{picture}}  
  \etb:U \otimes U^\vee \xrightarrow{\tilde{d}_U} \one,
\\[2em]
  \raisebox{-8pt}{
  \begin{picture}(26,22)
   \put(0,0){\scalebox{.75}{\includegraphics{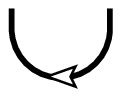}}}
   \put(0,0){
     \setlength{\unitlength}{.75pt}\put(-146,-155){
     \put(143,183)  {\scriptsize $ U $}
     \put(169,183)  {\scriptsize $ U^\vee $}
     }\setlength{\unitlength}{1pt}}
  \end{picture}}  
  \etb: \one \xrightarrow{\tilde d_U} U \otimes U^\vee,
  &
  \raisebox{-8pt}{
  \begin{picture}(26,22)
   \put(0,0){\scalebox{.75}{\includegraphics{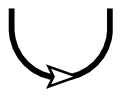}}}
   \put(0,0){
     \setlength{\unitlength}{.75pt}\put(-146,-155){
     \put(143,183)  {\scriptsize $ U^\vee $}
     \put(169,183)  {\scriptsize $ U $}
     }\setlength{\unitlength}{1pt}}
  \end{picture}}  
  \etb: \one \xrightarrow{\tilde{b}_U} U^\vee \otimes U,
  \nonumber
\end{array}
\ee
The twist is denoted by $\theta_U : U \rightarrow U$. For $f : U \rightarrow U$, the trace is defined as $\mathrm{tr}(f) = \tilde d_U \circ (f \otimes \id_{U^\vee}) \circ b_U \in \End(\one)$.

\begin{defn}[\cite{Turaev:1994,baki}] \rm
A {\em modular category} is a ribbon category, which is $\Cb$-linear, abelian, semi-simple, which has a simple tensor unit, and a finite number of isomorphism classes of simple objects. If $\{ U_i | i \in \Ic \}$ denotes a choice of representatives for these classes, in addition the complex $|\Ic|{\times}|\Ic|$-matrix $s_{i,j}$ defined by $s_{i,j} \, \id_\one = \mathrm{tr}(c_{U_i,U_j} \circ c_{U_j,U_i})$ is invertible.
\end{defn}

\begin{rem} \rm
For a VOA $V$ which satisfies the reductiveness and finiteness conditions stated in \cite{Huang:2005hh}, it is proved in \cite[Thm.\,4.6]{Huang:2005hh} that the category $\Rep V$ of $V$-modules is modular. We will refer to a VOA satisfying these conditions as {\em rational}.
\end{rem}

Let $\Cc$ be a modular category. The {\em dimension} of $U \in \Cc$ is defined as $\dim(U) \id_\one = \mathrm{tr}(\id_U)$, and the {\em global dimension} of $\Cc$ is defined to be
\be
  \Dim\Cc = \sum_{i \in \Ic} (\dim U_i)^2 \ .
\ee
The dimensions $\dim(U_i)$ of the simple objects are non-zero and real \cite[Thm.\,2.3\,\&\,Prop.\,2.9]{eno}, so that in particular $\Dim\Cc \ge 1$.

If $\Cc$ is a modular category, then $\bar\Cc$ denotes the modular category obtained from $\Cc$ by replacing braiding and twist by their inverses. Given two modular categories $\Cc$ and $\mathcal{D}$, denote by $\Cc \boxtimes \mathcal{D}$ their Deligne-product \cite{deligne,baki}, which in this case amounts to taking pairs of objects $U \boxtimes V$ and tensor products of Hom spaces, and completing with respect to direct sums. Every monoidal (and in particular every modular) category is equivalent to a strict one (which has trivial associator and unit isomorphisms). We will work with strict modular categories without further mention.

\subsection{Frobenius algebras and modular invariance}

The definitions given below only require some of the structure of a modular category, but rather than giving a minimal set of assumptions in each case, let us take $\Cc$ to be a modular category in this section.

\nxt An {\em algebra} in $\Cc$ is an object $A\in\Cc$ equipped with two morphisms $m_A: A\otimes A\rightarrow A$ and $\eta_A: \one \rightarrow A$ satisfying the usual associativity and unit properties (more details for this and the following can be found e.g.\ in \cite{Fuchs:2001qc}). 

\nxt An $A$-{\em left module} is an object $M \in \Cc$ equipped with a morphism $\rho_M : A \otimes M \rightarrow M$ compatible with unit and multiplication of $A$. Accordingly one defines right modules and bimodules, as well as intertwiners of modules.

\nxt A {\em coalgebra} is an object $A \in \Cc$ equipped with two morphisms $\Delta_A: A \rightarrow A\otimes A$ and $\epsilon: A \rightarrow \one$ satisfying the usual coassociativity and counit properties. 

\nxt A {\em Frobenius algebra} $A = (A,m,\eta,\Delta,\epsilon)$ is an algebra and a coalgebra such that
\be
  (\id_A \otimes m) \circ (\Delta \otimes \id_A) = \Delta \otimes m 
  = (m \otimes \id_A) \circ (\id_A \otimes \Delta) \ ,
\ee
i.e.\ the coproduct is an intertwiner of $A$-$A$-bimodules. 
We will use the following graphical representation for the morphisms of a Frobenius algebra:
\be
  m = \raisebox{-20pt}{
  \begin{picture}(30,45)
   \put(0,6){\scalebox{.75}{\includegraphics{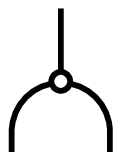}}}
   \put(0,6){
     \setlength{\unitlength}{.75pt}\put(-146,-155){
     \put(143,145)  {\scriptsize $ A $}
     \put(169,145)  {\scriptsize $ A $}
     \put(157,202)  {\scriptsize $ A $}
     }\setlength{\unitlength}{1pt}}
  \end{picture}}  
  ~,\quad
  \eta = \raisebox{-15pt}{
  \begin{picture}(10,30)
   \put(0,6){\scalebox{.75}{\includegraphics{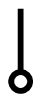}}}
   \put(0,6){
     \setlength{\unitlength}{.75pt}\put(-146,-155){
     \put(146,185)  {\scriptsize $ A $}
     }\setlength{\unitlength}{1pt}}
  \end{picture}}
  ~,\quad
  \Delta = \raisebox{-20pt}{
  \begin{picture}(30,45)
   \put(0,6){\scalebox{.75}{\includegraphics{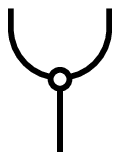}}}
   \put(0,6){
     \setlength{\unitlength}{.75pt}\put(-146,-155){
     \put(143,202)  {\scriptsize $ A $}
     \put(169,202)  {\scriptsize $ A $}
     \put(157,145)  {\scriptsize $ A $}
     }\setlength{\unitlength}{1pt}}
  \end{picture}}
  ~,\quad
  \epsilon = \raisebox{-15pt}{
  \begin{picture}(10,30)
   \put(0,10){\scalebox{.75}{\includegraphics{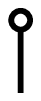}}}
   \put(0,10){
     \setlength{\unitlength}{.75pt}\put(-146,-155){
     \put(146,145)  {\scriptsize $ A $}
     }\setlength{\unitlength}{1pt}}
  \end{picture}}
  ~.
  \nonumber
\ee
A Frobenius algebra $A$ in $\Cc$ is called
\begin{itemize}
\item {\em haploid} iff $\dim \Hom(\one,A) = 1$,
\item {\em simple} iff it is simple as a bimodule over itself,
\item {\em special} iff  $m \circ \Delta = \zeta \, \id_A$ and $\epsilon \circ \eta = \xi \, \id_\one$ for nonzero constants $\zeta$, $\xi \in \mathbb{C}$,
\item {\em symmetric} iff
\be
  \raisebox{-35pt}{
  \begin{picture}(50,75)
   \put(0,8){\scalebox{.75}{\includegraphics{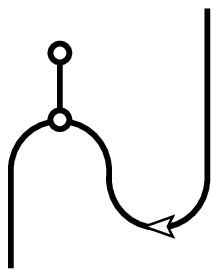}}}
   \put(0,8){
     \setlength{\unitlength}{.75pt}\put(-34,-37){
     \put(31, 28)  {\scriptsize $ A $}
     \put(87,117)  {\scriptsize $ A^\vee $}
     }\setlength{\unitlength}{1pt}}
  \end{picture}}
  ~=~
  \raisebox{-35pt}{
  \begin{picture}(50,75)
   \put(0,8){\scalebox{.75}{\includegraphics{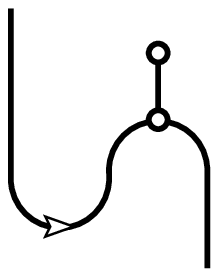}}}
   \put(0,8){
     \setlength{\unitlength}{.75pt}\put(-34,-37){
     \put(87, 28)  {\scriptsize $ A $}
     \put(31,117)  {\scriptsize $ A^\vee $}
     }\setlength{\unitlength}{1pt}}
  \end{picture}}
  ~~.
  \nonumber
\ee
\item {\em commutative} iff $m \circ c_{A,A} = m$,
\item {\em maximal} iff $\dim A = (\dim\Cc)^{\frac12}$, provided $A$ is also
haploid and commutative,
\item {\em modular invariant} iff $\theta_A = \id_A$ and for all $W \in \Cc$ we have
\be   
  \raisebox{-55pt}{
  \begin{picture}(102,120)
   \put(0,8){\scalebox{.75}{\includegraphics{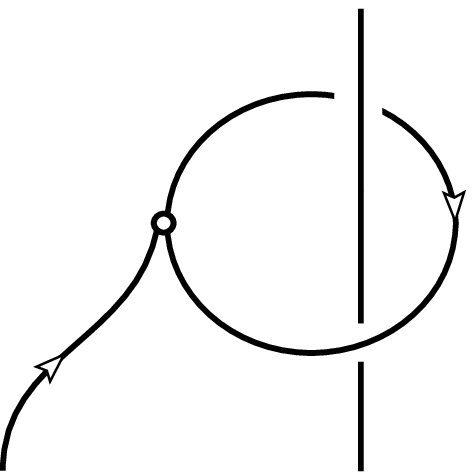}}}
   \put(0,8){
     \setlength{\unitlength}{.75pt}\put(-32,-25){
     \put( 27, 15)  {\scriptsize $ A $}
     \put( 88,133)  {\scriptsize $ A $}
     \put(130, 15)  {\scriptsize $W$ }
     \put(130,165)  {\scriptsize $W$ }
     }\setlength{\unitlength}{1pt}}
  \end{picture}}
\quad = ~ \sum_{i \in \Ic} \frac{\dim(U_i)}{(\dim\Cc)^{\frac12}} \quad 
 \raisebox{-55pt}{
  \begin{picture}(110,120)
   \put(0,8){\scalebox{.70}{\includegraphics{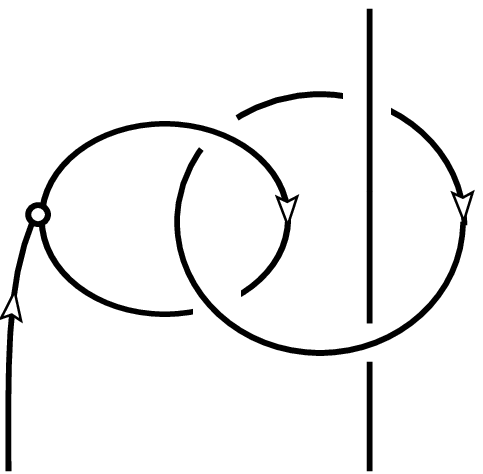}}}
   \put(0,8){
     \setlength{\unitlength}{.70pt}\put(-30,-25){
     \put( 27, 15)  {\scriptsize $ A $}
     \put( 45,124)  {\scriptsize $ A $}
     \put(130, 15)  {\scriptsize $W$ }
     \put(130,165)  {\scriptsize $W$ }
     \put(108,139)  {\scriptsize $ U_i $}
}\setlength{\unitlength}{1pt}}
  \end{picture}}
  ~~ .
\ee  
\end{itemize}
All the special symmetric
Frobenius algebras that will appear here are in fact `normalised' special in the sense that $\zeta =1$, which then implies $\xi = \dim(A)$. We will not mention the qualifier `normalised' explicitly below.

As an aside, we note that the name `maximal' is motivated as follows. A $A$-left module
$M$ is called {\em local} iff $\rho_M \circ c_{M,A}\circ c_{A,M} = \rho_M$ 
(see \cite{Kirillov:2001ti} or \cite[Sect.\,3.4]{Frohlich:2003hm}). We call a commutative algebra {\em maximal} iff its category of local modules is monoidally equivalent to the category of vector spaces. If a commutative maximal algebra $A$ is contained in another commutative algebra $B$ as a subalgebra, then $B$ is isomorphic to a direct sum of copies of $A$ as an $A$-module. Thus, if $A$ is haploid, it cannot be a subalgebra of a larger commutative haploid algebra.
In this sense, $A$ is `maximal'.
If $A$ is a haploid commutative Frobenius algebra of non-zero dimension, then $A$ is maximal iff $\dim(A) = (\Dim\Cc)^{\frac12}$ \cite[Thm.\,4.5]{Kirillov:2001ti}, hence the simplified definition above.

The modular invariance condition above is the least standard (and the most complicated) notion. It was introduced in \cite{cardy-cond} (see \cite[Lem.\,3.2]{Kong:2008ci} for the relation to the definition above), and we included it for the sake of Remark~\ref{rem:full-field} below. Fortunately, for the case of interest to us it can be replaced by a much simpler condition:

\begin{thm}[{\cite[Thm.\,3.4]{Kong:2008ci}}]
Let $A$ be a haploid commutative symmetric Frobenius algebra in $\Cc$. Then $A$ is modular invariant iff it is maximal. In either case, $A$ is in addition special.
\end{thm}

\begin{rem} \label{rem:full-field} \rm
There are many approaches to axiomatise properties of conformal field theories, see e.g.\ \cite{Belavin:1984vu,Friedan:1986ua,Borcherds:1986,Vafa:1987ea,FLM,Moore:1988qv,Segal,huang:91,Gaberdiel:1998fs,Kapustin:2000aa}.
We will use those developed in \cite{Huang:2005hk,Huang:2006ar,cardy-cond} and \cite{Fuchs:2001am,tft5,Fjelstad:2006aw}. Let $V_L$ and $V_R$ be two rational VOAs such that $c_L - c_R \equiv 0 \mod 24$. A CFT over $V_L \otimes_\Cb V_R$ in the sense of Section~\ref{sec:d+d}, is -- in the nomenclature of \cite{Huang:2006ar, cardy-cond} -- a conformal full field algebra over $V_L \otimes_\Cb V_R$ with non-degenerate invariant bilinear form, which is modular invariant and has a unique vacuum. Let $\Cc_L = \Rep V_L$ and $\Cc_R = \Rep V_R$.
It is shown in \cite[Thm.\,6.7]{cardy-cond} that CFTs over $V_L \otimes_\Cb V_R$ are in one-to-one correspondence with haploid commutative symmetric Frobenius algebras in $\Cc_L \boxtimes \bar\Cc_R$ which are modular invariant.  
\end{rem}

\subsection{The full centre}

Fix a modular category $\Cc$. The braiding on $\Cc$ allows to endow the functor $T : \CxC \rightarrow \Cc$, given by the tensor product on $\Cc$, with the structure of a tensor functor. This can be done in two ways, and we choose the convention of \cite[Sect.\,2.4]{Kong:2008ci}. 
The functor $T$ has an adjoint $R : \Cc \rightarrow \CxC$, that is, there is a bi-natural family of isomorphisms
\be
  \hat\chi_{Y,V} : 
  \Hom_{\Cc}(T(Y), V) \longrightarrow \Hom_{\CxC}(Y, R(V)) \ .
\ee
In fact, $R$ is both left and right adjoint to $T$, but we will not need this. Denote the two natural transformations associated to the adjunction by
\be
\id_{\CxC} 
\xrightarrow{\hat{\delta}} RT  
\qquad  \text{and} \qquad
TR \xrightarrow{\hat{\rho}} \id_{\Cc} \ .
\ee
They are $\hat{\delta}_Y = \hat{\chi}(\id_{T(Y)})$ and
$\hat{\rho}_V = \hat{\chi}^{-1}(\id_{R(V)})$
for $V\in \Cc$, $Y \in \CxC$.
Explicit expressions for $\hat\chi$, $\hat\delta$ and $\hat\rho$ are given in \cite[Sect.\,2.4]{Kong:2008ci}.
The functor $R$ obeys
\be \label{eq:R-on-obj}
  R(\one) = \bigoplus_{i \in \Ic} U_i^\vee \boxtimes U_i
  \quad , \quad
  R(V) \cong (V\boxtimes \one) \otimes R(\one) \ .
\ee

\begin{prop}[{\cite[Prop.\,2.16,\,2.24,\,2.25]{Kong:2008ci}}] \label{prop:T-R-prop}
Let $A \in \Cc$ and $B \in \CxC$ be algebras.
\\[.3em]
(i) If $A$ and $B$ are special symmetric Frobenius, so are $R(A) \in \CxC$ and $T(B) \in \Cc$.
\\[.3em]
(ii) A morphism $f : T(B) \rightarrow A$ is an algebra homomorphism iff
$\hat\chi(f) : B \rightarrow R(A)$ is  an algebra homomorphism.
\end{prop}

The structure morphisms for $R(A)$ and $T(B)$ in part (i) are given in \cite[Sect.\,2.2]{Kong:2008ci}. Part (ii) shows in particular that $\hat\rho_A : TR(A) \rightarrow A$ is an algebra map.

\medskip

For an algebra in a braided category one can define a left and a right centre \cite{vz,Ostrik:2001}. We will only need the left centre. Given an algebra $A$ in $\Cc$, its {\em left centre} $C_l(A) \hookrightarrow A$ is the largest subobject of $A$ such that the composition
\be
  C_l(A) \otimes A \rightarrow A \otimes A \xrightarrow{c_{A,A}} A \otimes A \xrightarrow{m_A} A
\ee
coincides with the composition
\be
  C_l(A) \otimes A \rightarrow A \otimes A \xrightarrow{m_A} A \ .
\ee
If $A$ is special symmetric Frobenius (and $\Cc$ abelian), the left centre exists and can be written as the image of an idempotent defined in terms of $m_A$, $\Delta_A$, $c_{A,A}$ and the duality morphisms, see \cite[Sect.\,2.4]{Frohlich:2003hm} for details.

\begin{defn}[{\cite[Def.\,4.9]{Fjelstad:2006aw}}] \rm
The {\em full centre} of a special symmetric Frobenius algebra 
$A$ in a modular category $\Cc$
is $Z(A) = C_l(R(A)) \in \CxC$.
\end{defn}

The full centre has a natural generalisation to algebras in general monoidal categories, in which case it provides a commutative algebra in the monoidal centre of the category and is characterised by a universal property \cite{da1}.

Denote the subobject embedding and restriction morphisms by
\be
  e_Z: Z(A) \hookrightarrow R(A) 
  \qquad \text{and} \qquad
  r_Z: R(A) \twoheadrightarrow Z(A)
  \ .
  \label{eq:e-r-RA-ZA-def}
\ee
They obey $r_Z\circ e_Z = \id_{Z(A)}$, i.e.\ $Z(A)$ is a direct summand of $R(A)$. By construction of the algebra structure on $Z(A)$, the map $e_Z$ is an algebra homomorphism.

\begin{thm}[{\cite[Prop.\,2.7]{Kong:2007yv} and \cite[Thm.\,3.22]{Kong:2008ci}}]
\label{thm:full-centre-properties}
Let $\Cc$ be a modular category.\\[.3em]
(i) The full centre of a simple special symmetric Frobenius algebra in $\Cc$ is a haploid commutative maximal special symmetric Frobenius algebra in $\CxC$. 
\\[.3em]
(ii) Every haploid commutative maximal special symmetric Frobenius algebra in $\CxC$ is isomorphic as an algebra to the full centre of some simple special symmetric Frobenius algebra in $\Cc$.
\end{thm}

\subsection{Bimodules and defects}

Fix a modular category $\Cc$. Let $A,B,C$ be algebras in $\Cc$. An $A$-$B$-{\em bimodule} $X$ is an $A$-left module and a $B$-right module such that the two actions commute. Given a $B$-$C$-bimodule $Y$, we define the $A$-$C$-bimodule $X \otimes_B Y$ as a cokernel in the usual way. If $B$ is special symmetric Frobenius, $X \otimes_B Y$ can be written as the image of an idempotent on $X \otimes Y$, and so in this case $X \otimes_B Y \hookrightarrow X \otimes Y$ is a
direct summand (as a bimodule).
We denote the embedding and restriction maps as
\be \label{eq:tensor-proj}
  e_B : X \otimes_B Y \hookrightarrow X \otimes Y
  \quad , \quad
  r_B : X \otimes Y \twoheadrightarrow X \otimes_B Y \ ,
\ee
such that $r_B \circ e_B = \id_{X \otimes_B Y}$. To keep the notation at bay, we will not include labels for $X$ and $Y$.

\begin{rem} \rm
In the approach to CFT correlators via three-dimensional topological field theory given in \cite{Fuchs:2001am,tft5,Fjelstad:2006aw}, a CFT is specified by a special symmetric Frobenius algebra $A$ in $\Rep V$. In this approach, one automatically obtains an open/closed CFT which satisfies genus 0 and genus 1 consistency conditions (and, subject to modular functor properties of higher genus conformal blocks, is in fact well-defined on surfaces of arbitrary genus). The bulk CFT one finds in this way is the CFT over $V \otimes_\Cb V$ described by $Z(A)$ via Remark~\ref{rem:full-field}, see \cite[Sect.\,4.3]{Fjelstad:2006aw}. 
\\
In the TFT approach, one can also describe CFTs in the presence of topological defect lines  which respect the $V \otimes_\Cb V$ symmetry \cite{Frohlich:2006ch}. Different patches of the CFT world sheet are labelled by special symmetric Frobenius algebras and the defects (or domain walls) between them by bimodules. The fusion of defect lines translates into the tensor product of bimodules over their intermediate algebra. In this way, CFTs over $V \otimes_\Cb V$ become a bicategory \cite{Schweigert:2006af}, where objects are CFTs, 1-morphisms are topological defects preserving $V \otimes_\Cb V$, and 2-morphisms are `defect fields' in the vacuum representation (described by intertwiners of bimodules).
\end{rem}

\subsection{Equivalence of groupoids}

\begin{defn} \label{def:group-P-C} 
\rm
Let $\Cc$ be a modular category.
\\[.3em]
(i) $\mathbf{P}(\Cc)$ is the groupoid whose objects are simple special symmetric Frobenius algebras $A,B,\dots$ in $\Cc$ and whose morphisms $A \rightarrow B$ are isomorphism classes of invertible $B$-$A$-bimodules.
\\[.3em]
(ii) $\mathbf{A}(\Cc)$ is the groupoid whose objects are simple commutative maximal special symmetric Frobenius algebras $C,D,\dots$ in $\Cc$ and whose morphisms $C \rightarrow D$ are algebra isomorphisms from $C$ to $D$.
\end{defn}

In the remainder of this section we will prove the statement announced in the introduction, namely that the two groupoids $\mathbf{P}(\Cc)$ and $\mathbf{A}(\CxC)$ are equivalent (Theorem~\ref{thm:main} below).
The proof will be split into several lemmas. We start by constructing a functor $Z : \mathbf{P}(\Cc) \rightarrow  \mathbf{A}(\CxC)$. On objects it is given by taking the full centre (hence the notation `$Z$'),
\be
  Z(A) = C_l(R(A)) 
  \qquad \text{for} ~~ A \in \mathbf{P}(\Cc)  \ .
\ee
In order to define the functor $Z$ on morphisms, we need some more notation. Fix two objects $A,B \in \mathbf{P}(\Cc)$, i.e.\ two simple special symmetric Frobenius algebras. 
Given a $B$-$A$-bimodule $X$, we define a morphism $\phi_X : Z(A) \rightarrow Z(B)$ as in 
\cite{Fuchs:2007vk} and \cite[Lem.\,3.2]{Kong:2007yv},
\be
 \phi_X = \frac{\dim(X)}{\dim(A)} \,\, \,\, 
 \raisebox{-65pt}{
  \begin{picture}(100,140)
   \put(0,8){\scalebox{.75}{\includegraphics{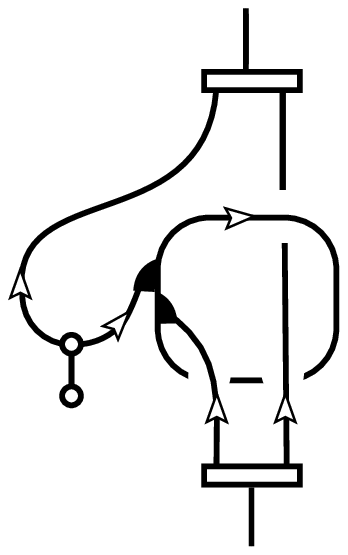}}}
   \put(0,8){
     \setlength{\unitlength}{.75pt}\put(-290,-160){
     \put(337,237)  {\scriptsize $ X{\boxtimes} \one $}
     \put(309,273)  {\scriptsize $ B {\boxtimes} \one $}
     \put(316,188)  {\scriptsize $ A {\boxtimes} \one $}
     \put(376,273)  {\scriptsize $ R(\one) $}
     \put(350,322)  {\scriptsize $ Z(B) $}
     \put(350,150)  {\scriptsize $ Z(A) $}
     \put(380,181)  {\scriptsize $ e_Z $}
     \put(380,294)  {\scriptsize $ r_Z $}
  }\setlength{\unitlength}{1pt}}
  \end{picture}}   ,
\label{eq:ZX-def}
 \ee
where $e_Z$ and $r_Z$ have been introduced in \eqref{eq:e-r-RA-ZA-def}.
We define the functor $Z$ on morphisms of $\mathbf{P}(\Cc)$ as
\be
  Z(X) = \phi_X
  \qquad \text{for} ~~ X : A \rightarrow B \ .
\ee
The following lemma implies that $Z$ is well-defined and functorial.

\begin{lem}[{\cite[Lem.\,3.1,\,3.2,\,3.3]{Kong:2007yv}}]\label{ZX-properties}
Let $A, B, C$ be simple special symmetric Frobenius algebras in $\Cc$ and let $X,X'$ be $C$-$B$-bimodules and $Y$ a $B$-$A$-bimodule. Then
\\[.3em]
(i) If $X \cong X'$ as bimodules, then $\phi_X = \phi_{X'}$. 
\\[.3em]
(ii) $\phi_A = \id_{Z(A)}$.
\\[.3em]
(iii) $\phi_X \circ \phi_Y = \phi_{X\otimes_B Y}$. 
\\[.3em]
(iv) If $X^\vee \otimes_B X \cong A$ or $X\otimes_A X^\vee\cong B$ as bimodules, then $\phi_X$ is an algebra isomorphism. 
\end{lem}

In the following we will give a series of lemmas which will show that the functor $Z$ is full, faithful and essentially surjective.

\medskip

Let $\Jc_A$ be a label set for the isomorphism classes of simple left $A$-modules and let $\{ M_\kappa | \kappa \in \Jc_A \}$ be a choice of representatives. Define
\be
  T_A = \bigoplus_{\lambda \in \mathcal{J_A}} M_{\lambda}^\vee \otimes_A M_{\lambda} \ ,
\ee
Each of the objects $M_{\lambda}^\vee \otimes_A M_{\lambda}$ is naturally a haploid algebra in $\Cc$ (see e.g.\ \cite[Lem.\,4.2]{Kong:2007yv}), and thus also $T_A$ is an algebra (non-haploid in general).
Define the morphisms
\begin{align} \label{eq:iota-pi-e-r}
  \iota_\kappa &: M_{\kappa}^\vee \otimes_A M_{\kappa} \hookrightarrow T_A \ ,
& 
e_\kappa &:= T_A \overset{\pi_\kappa}{\twoheadrightarrow}
  M_\kappa^\vee \otimes_A M_\kappa
  \overset{e_A}{\hookrightarrow} M_\kappa^\vee \otimes M_\kappa \ ,
\nonumber\\
  \pi_\kappa &: T_A \twoheadrightarrow M_{\kappa}^\vee \otimes_A M_{\kappa} \ ,
&
  r_\kappa &:= M_\kappa^\vee \otimes M_\kappa 
\overset{r_A}{\twoheadrightarrow} 
M_\kappa^\vee \otimes_A M_\kappa \overset{\iota_\kappa}{\hookrightarrow} T_A \ ,
\end{align}
where $e_A$ and $r_A$ where given in \eqref{eq:tensor-proj}.
Note that by definition of the algebra structure on $T_A$, $\pi_\kappa$ is an algebra map, while
$\iota_\kappa$ respects the multiplication but not necessarily the unit.

From Proposition~\ref{prop:T-R-prop}
we know that $T(Z(A))$ is a special symmetric Frobenius algebra (because $A$ is), and 
from \eqref{eq:R-on-obj} we have
$T(R(A)) = \bigoplus_{i \in \Ic} A \otimes U_i^\vee \otimes U_i$. Using
the maps \eqref{eq:e-r-RA-ZA-def} we can define 
\begin{align}
e_i &= TZ(A) \overset{T(e_Z)}\hookrightarrow TR(A)
\twoheadrightarrow A \otimes U_i^\vee \otimes U_i
\nonumber \\
r_i &= A \otimes U_i^\vee \otimes U_i \hookrightarrow TR(A) 
\overset{T(r_Z)}{\twoheadrightarrow} TZ(A)
\ .
\end{align}
Using these ingredients we
define two morphisms $\varphi: TZ(A) \rightarrow T_A$ 
and $\bar{\varphi}: T_A \rightarrow TZ(A)$ by 
\be
\varphi ~=~ \sum_{i \in \Ic} \sum_{\kappa \in \Jc_A}
  \raisebox{-58pt}{
  \begin{picture}(64,120)
   \put(0,8){\scalebox{.75}{\includegraphics{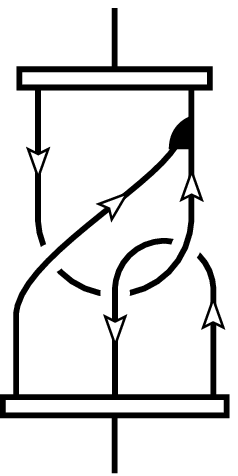}}}
   \put(0,8){
     \setlength{\unitlength}{.75pt}\put(-18,-0){
     \put(77,86)  {\scriptsize $ M_\kappa $}
     \put(38,32)  {\scriptsize $ U_i $}
     \put(45,85)  {\scriptsize $ A $}
     \put(83,115)  {\scriptsize $ r_\kappa $}
     \put(87,19)  {\scriptsize $ e_i $}
     \put(45,140)  {\scriptsize $ T_A $}
     \put(40,-9)  {\scriptsize $ TZ(A) $}
     }\setlength{\unitlength}{1pt}}
  \end{picture}} 
\quad , \quad
\bar\varphi ~=~ \sum_{i \in \Ic} \sum_{\kappa \in \Jc_A}
\frac{\dim(U_i)\dim(M_\kappa)}{\Dim \Cc}
  \raisebox{-58pt}{
  \begin{picture}(64,120)
   \put(0,8){\scalebox{.75}{\includegraphics{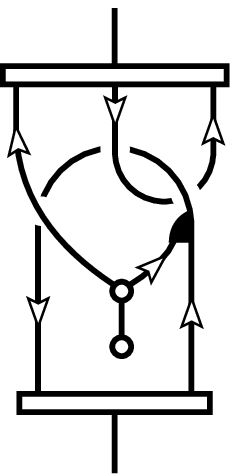}}}
   \put(0,8){
     \setlength{\unitlength}{.75pt}\put(-18,-0){
     \put(77,56)  {\scriptsize $ M_\kappa $}
     \put(56,101)  {\scriptsize $ U_i $}
     \put(47,65)  {\scriptsize $ A $}
     \put(87,115)  {\scriptsize $ r_i $}
     \put(83,19)  {\scriptsize $ e_\kappa $}
     \put(40,140)  {\scriptsize $ TZ(A) $}
     \put(45,-9)  {\scriptsize $ T_A $}
     }\setlength{\unitlength}{1pt}}
  \end{picture}} 
  .
\labl{eq:phi-phi-def}
It has been shown in \cite[Prop.\,4.3\,\&\,Lem.\,4.6,\,4.7]{Kong:2007yv} that $\varphi$ and $\bar\varphi$ are inverse to each other, and that they are algebra isomorphisms.

\medskip

Fix another simple special symmetric Frobenius algebra $B$ and let $\rho : M_\lambda \otimes B \rightarrow M_\lambda$ be a right $B$-action on a simple left $A$-module $M_\lambda$ which commutes with the left $A$-action. Denote the resulting $A$-$B$-bimodule by $M_\lambda(\rho)$ and define the morphism
\be \label{eq:g-lam-rho-def}
  g_{\lambda}(\rho) = \frac{\dim(M_\lambda)}{\dim(A)} 
\raisebox{-45pt}{
  \begin{picture}(80,100)
   \put(15,8){\scalebox{.75}{\includegraphics{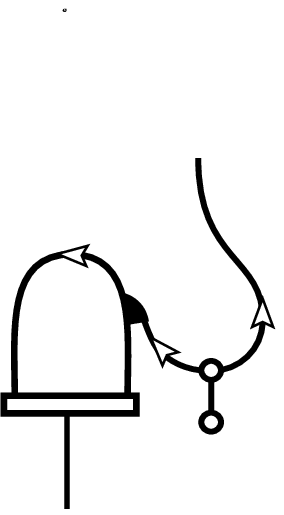}}}
   \put(15,8){
     \setlength{\unitlength}{.75pt}\put(-18,-11){
     \put(12, 0) {\scriptsize $M_{\lambda}^\vee \otimes_A M_{\lambda}$}
     \put(70, 118)     {\scriptsize $B$} 
     \put(64, 64)   {\scriptsize $B$}
     \put(32, 92)   {\scriptsize $M_{\lambda}(\rho)$ }
     \put(2, 40)   {\scriptsize $e_A$ }
     }\setlength{\unitlength}{1pt}}
  \end{picture}} 
  \quad . 
\ee
One quickly checks that $g_{\lambda}(\rho)$ is an intertwiner of $B$-$B$-bimodules.

\begin{lem} \label{lem:Z-via-adjunct}
The following equality of morphisms $Z(A)\rightarrow Z(B)$ holds:
\be  \label{eq:phi=}
\phi_{M_{\lambda}(\rho)^\vee} = r_Z \circ \hat{\chi}(  g_{\lambda}(\rho) \circ \pi_{\lambda} \circ \varphi)
 \ .
\ee
\end{lem}

\begin{proof}
The identity can be established by
composing the graphical expressions of $\varphi$, $g_{\lambda}(\rho)$ and $\hat{\chi}$ (see \cite[Eqn.\,(2.43)]{Kong:2008ci}) and comparing the result to the graphical expression \eqref{eq:ZX-def} for $\phi_X$.
\end{proof}

\begin{lem} \label{lem:Mlf-prop}
Let $A,B \in \mathbf{P}(\Cc)$ be haploid.
Given an algebra isomorphism $f: Z(A) \rightarrow Z(B)$, there exist $\lambda_f \in \Jc_A$ and a right $B$-action $\rho_f$ on $M_{\lambda_f}$ such that
\\[.3em]
(i) $M_{\lambda_f}(\rho_f)$ is an invertible $A$-$B$-bimodule,
\\[.3em]
(ii) $Z(M_{\lambda_f}(\rho_f)^\vee) = f$.
\end{lem}

\begin{proof}
Given the algebra isomorphism $f$ and an index $\lambda \in \Jc_A$, we can define the map
\be  \label{eq:map-h}
h(f,\lambda) : M_{\lambda}^\vee \otimes_A M_{\lambda} 
\overset{\iota_\lambda}{\hookrightarrow} T_A  \xrightarrow{\bar{\varphi}} TZ(A) 
\xrightarrow{T(f)} TZ(B) \xrightarrow{Te_Z} TR(B) \xrightarrow{\hat{\rho}} B\ .  
\ee
It is shown in part e) of the proof of \cite[Thm.\,1.1]{Kong:2007yv} that there exists a unique $\lambda_f \in \Jc_A$ such that $h(f,\lambda_f) \neq 0$. We have already seen that all the individual maps above respect the algebra multiplication. The map $\iota_\kappa$ does in general not preserve the unit, but because $M_{\lambda}^\vee \otimes_A M_{\lambda}$ and $B$ are haploid, the composite map $h(f,\lambda)$ does. This amounts to the argument in part b) and e) of the proof of \cite[Thm.\,1.1]{Kong:2007yv}, which shows that $h(f,\lambda_f)$ is an algebra isomorphism. 
We can use the isomorphism $h(f,\lambda_f)$ to define a right $B$-action on $M \equiv M_{\lambda_f}$ by setting
\be
  \rho_f : M \otimes B
  \xrightarrow{\id \otimes h(f,\lambda_f)^{-1}}
  M \otimes M^\vee \otimes_A M
  \xrightarrow{\id \otimes e_A}
  M \otimes M^\vee \otimes M
  \xrightarrow{\tilde d_M \otimes id}
  M
\ee  
By construction, we now have $h(f,\lambda_f) :  M^\vee \otimes_A M \xrightarrow{\,\sim\,} B$ 
as $B$-$B$-bimodules, which implies that $M$ is an invertible $A$-$B$-bimodule (see e.g.\ \cite[Lem.\,3.4]{Frohlich:2006ch}). This proves part (i).

Let us now turn to part (ii). We first claim that
\be \label{eq:g-lam=h-lam}
  g_{\lambda_f}(\rho_f) = h(f,\lambda_f) \ .
\ee
To see this identity first note that both sides are intertwiners of $B$-$B$-bimodules. Furthermore, $M^\vee \otimes_A M$ and $B$ are both simple as $B$-$B$-bimodules (because $B$ is simple). Thus $g_{\lambda_f}(\rho_f) = \xi h(f,\lambda_f)$ for some $\xi \in \Cb$. To determine $\xi$ we let both sides act on the unit $r_A \circ \tilde b_M$ of $M^\vee \otimes_A M$. As $h(f,\lambda_f)$ is an algebra map, it gives $\eta_B$. For the left hand side one uses the explicit form \eqref{eq:g-lam-rho-def} together with \cite[Lem.\,3.3\,\&\,Eqns.\,(3.4),\,(3.7)]{Kong:2007yv} to find that it is also equal to $\eta_B$. Thus $\xi=1$.

Next consider the equalities
\begin{align} \label{eq:h-pi-f}
h(f,\lambda_f) \circ \pi_{\lambda_f} \circ \varphi  
&\overset{(1)}{=} \sum_{\lambda \in \Jc_A} h(f,\lambda) \circ \pi_{\lambda} \circ \varphi
\overset{(2)}{=} \sum_{\lambda \in \Jc_A} \hat{\rho} \circ T(e_Z\circ f) \circ \bar{\varphi} \circ \iota_{\lambda} \circ \pi_{\lambda} \circ \varphi  \nn
&\overset{(3)}{=} \hat{\rho} \circ T(e_Z \circ f) 
\overset{(4)}{=} \hat{\chi}^{-1} (e_Z \circ f) \ . 
\end{align}
Step (1) uses that $h(f,\lambda)$ is only non-zero for $\lambda = \lambda_f$, 
in step(2) we inserted the definition \eqref{eq:map-h} of $h(f,\lambda)$, and
step (3) amounts to the identity $\sum_\lambda \iota_{\lambda} \circ \pi_{\lambda} = \id_{T_A}$ and the fact that $\bar\varphi$ is the inverse of  $\varphi$. Finally, step (4) follows from the definition of $\hat\rho$ and naturality of $\hat\chi$, see  \cite[Eqn.\,(2.53)]{Kong:2008ci}.
By Lemma~\ref{lem:Z-via-adjunct}, \eqref{eq:g-lam=h-lam} and \eqref{eq:h-pi-f} we have
\be
Z(M_{\lambda_f}(\rho_f)^\vee)
= r_Z \circ \hat{\chi}(  h(f,\lambda_f) \circ \pi_{\lambda_f} \circ \varphi)
= r_Z \circ \hat{\chi}( \hat{\chi}^{-1} (e_Z \circ f) ) = f \ .
\ee
This shows part (ii).
\end{proof}

\begin{lem} \label{lem:faithful}
Let $A,B \in \mathbf{P}(\Cc)$ be haploid.
Let $X$ be an invertible $B$-$A$-bimodule and let $f = Z(X) : Z(A) \rightarrow Z(B)$ be the corresponding algebra isomorphism (Lemma~\ref{ZX-properties}\,(iv)). Choose $\lambda_f$ and $\rho_f$ as in Lemma~\ref{lem:Mlf-prop}. Then $X \cong M_{\lambda_f}(\rho_f)^\vee$ as $B$-$A$-bimodules.
\end{lem}

\begin{proof}
Since $X$ is invertible, it is necessarily simple as a $B$-$A$-bimodule (see e.g. \cite[Lem.\,3.4]{Frohlich:2006ch}). In fact, it is even simple as a right $A$-module, because, if $X \cong M \oplus N$ as right $A$-modules, then $X \otimes_A X^\vee = M \otimes_A M^\vee \oplus N \otimes_A N^\vee \oplus \cdots$ would not be haploid. But $X \otimes_A X^\vee \cong B$, as $X$ is invertible, and so $X \otimes_A X^\vee$ is haploid. 

Thus there is a $\lambda_0 \in \Jc_A$ such that $M_{\lambda_0} \cong X^\vee$ as left $A$-modules. We will now show that $\lambda_0 = \lambda_f$. Denote by $\rho$ the right $B$-action on $M_{\lambda_0} $ induced by the isomorphism $M_{\lambda_0} \cong X^\vee$. By Lemma~\ref{lem:Z-via-adjunct} we have $Z(X) =  r_Z \circ \hat{\chi}(g_{\lambda_0}(\rho) \circ \pi_{\lambda_0} \circ \varphi)$. Then, 
\begin{align}
h(Z(X),\kappa) 
&\overset{(1)}= \hat{\rho} \circ T(e_Z \circ Z(X)) \circ \bar{\varphi} \circ \iota_{\kappa} \nn
&\overset{(2)}= \hat{\chi}^{-1} (e_Z \circ r_Z \circ \hat{\chi}(g_{\lambda_0}(\rho) \circ \pi_{\lambda_0} \circ \varphi) ) \circ \bar{\varphi} \circ \iota_{\kappa}
\nn
&\overset{(3)}= \hat{\chi}^{-1} (\hat{\chi}(g_{\lambda_0}(\rho) \circ \pi_{\lambda_0} \circ \varphi) ) \circ \bar{\varphi} \circ \iota_{\kappa} \nn
&\overset{(4)}= g_{\lambda_0}(\rho) \circ \pi_{\lambda_0} \circ \iota_{\kappa} 
\overset{(5)}= \delta_{\lambda_0,\kappa} \, g_{\lambda_0}(\rho) \ ,     
\end{align}
where step (1) is the definition of the map $h$ from \eqref{eq:map-h}, in step (2) we inserted the expression for $Z(X)$ just obtained, step (3) amounts to \cite[Lem.\,3.1\,(iv)]{Kong:2007yv}, step (4) uses that $\bar\varphi$ is the inverse of $\varphi$, and step (5) is just the definition of the maps $\pi_{\lambda_0}$ and $\iota_{\kappa}$ in \eqref{eq:iota-pi-e-r}.

In the proof of Lemma~\ref{lem:Mlf-prop}, $\lambda_f$ is defined to be the unique element of $\Jc_A$ for which $h(f,\lambda)$ is non-zero. Thus the above calculation shows $\lambda_f = \lambda_0$. On the other hand, it follows from \eqref{eq:g-lam=h-lam} and the above calculation that
\be
  g_{\lambda_f}(\rho_f) = g_{\lambda_0}(\rho) \ .
\ee
This equality in turn implies that $\rho = \rho_f$, and thus the right $B$-actions on 
$M_{\lambda_f}(\rho_f)$ and $M_{\lambda_0} \cong X^\vee$ agree, i.e.\ $M_{\lambda_f}(\rho_f) \cong X^\vee$ as $A$-$B$-bimodules.
\end{proof}

We have now gathered the necessary ingredients to prove our main result.

\begin{thm} \label{thm:main}  
Let $\Cc$ be a modular category. The groupoids $\mathbf{P}(\Cc)$ and $\mathbf{A}(\CxC)$ given in Definition~\ref{def:group-P-C} are equivalent.
\end{thm}

\begin{proof}
By Theorem~\ref{thm:full-centre-properties}\,(ii) the functor $Z$ is essentially surjective. Fix two objects $A,B \in \mathbf{P}(\Cc)$. We need to show that $Z$ provides an isomorphism between the morphism spaces $A \rightarrow B$ and $Z(A) \rightarrow Z(B)$. 
By \cite[Prop.\,4.10]{Kong:2007yv} there exist haploid algebras
$A',B' \in \mathbf{P}(\Cc)$ and invertible bimodules $X : A \rightarrow A'$ and $Y : B \rightarrow B'$. It is thus enough to show that
\be
  Z(-) : \Hom_{\mathbf{P}(\Cc)}(A',B') \longrightarrow \Hom_{\mathbf{A}(\CxC)}(Z(A'),Z(B'))
\ee
is an isomorphism. By Lemma~\ref{lem:Mlf-prop}, $Z(-)$ is full, and by Lemma~\ref{lem:faithful}, it is faithful.
\end{proof}

\section{Examples}\label{sec:ex}

\subsection{Simple currents models}

Let $V$ be a rational VOA with the property that $\Cc = \Rep V$ is {\em pointed}, i.e.\ every simple object of $\Cc$ is invertible. In other words, $\Cc$ is generated by simple currents. A large class of examples of such VOAs are provided by lattice VOAs (see for example \cite{FLM}).
\medskip

A pointed braided monoidal category $\Cc$ is characterised by a finite abelian group $A$ (the group of simple currents) together with a quadratic function $q:A\to\Cb^*$ encoding their braid statistics \cite[Sect.\,3]{js}. $\Cc$ is modular if the quadratic function is {\em non-degenerate}, i.e.\ if the associated bi-multiplicative function $\sigma:A\times A\to\Cb^*$ defined by 
\be
  \sigma(a,b) = q(ab)q(a)^{-1}q(b)^{-1}
\ee 
is non-degenerate in the sense that for each $a \neq 1$ the homomorphism
$\sigma(a,-):A\to\Cb^*$ is non-trivial. 
 
The structure of a modular category is encoded in the group of (isomorphism classes of) simple objects $A$, a 3-cocycle $\alpha\in Z^3(A,\Cb^*)$, which controls the associativity constraint and a certain function $c:A\times A\to \Cb^*$, controlling the braiding (see \cite[Sect.\,3]{js} for the conditions on $c$). The pair $(\alpha,c)$ is known as an {\em abelian 3-cocycle} of $A$ with coefficients in $\Cb^*$. It was shown in \cite{em} that the group of classes of abelian 3-cocycles modulo coboundaries coincides with the group of quadratic functions. In other words up to a braided equivalence a pointed category depends only on the quadratic function $q$, defined by $q(a) = c(a,a)$ (see \cite[Sect.\,3]{js}). We will denote a representative of this class by  $\Cc(A,q)$.

Isomorphism classes of haploid special symmetric Frobenius algebras (also called {\em Schellekens algebras} in this context \cite[Def.\,3.7]{tft3}) are labelled by pairs $(B,\beta)$, where $B\subset A$ is a subgroup and $\beta:B\times B\to \Cb^*$ is a symmetric bi-multiplicative function such that $\beta(b,b) = q(b)$ for $b\in B$ \cite[Def.\,3.17,\,Prop.\,3.22]{tft3}. A Schellekens algebra corresponding to $(B,\beta)$ is commutative iff $\beta = 1$. This means that commutative Schellekens algebras correspond to {\em isotropic} subgroups (subgroups on which $q$ restricts trivially). 

The details of the following discussion will appear elsewhere.

A commutative Schellekens algebra is maximal iff the corresponding subgroup is maximal isotropic, i.e.\ {\em Lagrangian}. In particular, commutative maximal Schellekens algebras in $\Cc(A,q)\boxtimes \overline{\Cc(A,q)} = \Cc(A,q)\boxtimes \Cc(A,q^{-1}) = \Cc(A\times A,q\times q^{-1})$ correspond to subgroups in $A\times A$, Lagrangian with respect to $q\times q^{-1}$. The full centre of a Schellekens algebra $R = R(B,\beta)$ in $\Cc(A,q)$ for a pair $(B,\beta)$ corresponds to the Lagrangian subgroup
\be
  \Gamma(B,\beta) = \{ (a,a^{-1}b)|\ a\in A,\ b\in B,\ \mbox{such that}\ \sigma(c,a)=\beta(c,b)\ \forall c\in B\}
\ee  
in $A\times A$. The construction of $\Gamma$ gives an isomorphism between the set of pairs $(B,\beta)$ and the set of Lagrangian subgroups in $A\times A$. This also provides the isomorphism \cite[Cor.\,3.25]{Kong:2008ci} between the set of Morita classes of simple special symmetric Frobenius algebras in $\Cc(A,q)$ and the set of isomorphism classes of simple commutative maximal special symmetric Frobenius algebras in $\Cc(A,q)\boxtimes \overline{\Cc(A,q)}$).

The automorphism group of a Schellekens algebra $R$ corresponding to $(B,\beta)$ is the {\em dual} group $\hat B = \Hom(B,\Cb^*)$ (the group of characters). In particular the automorphism group of the full centre of $R$ is the group $\widehat {\Gamma(B,\beta)}$. This is in agreement with \cite[Prop.\,5.14]{Frohlich:2006ch}, where it was established that the group $\Pic(R)$ of isomorphism classes of invertible $R$-$R$-bimodules fits into a short exact sequence
\be\label{fes}
B \to \hat B\times A\to \Pic(R),
\ee
where the first map sends $b\in B$ into $(\beta(-,b)^{-1},b)$. It is easy to see that the group $\Gamma(B,\beta)$ fits into a short exact sequence
\be\label{ses}
\Gamma(B,\beta) \to A\times B \to \hat B ,
\ee
where the first map is $(u,v) \mapsto (u,uv)$, and the second map sends $(a,b)$ into $\sigma(-,a)\beta(-,b)^{-1}$. The sequence \eqref{fes} is isomorphic to the sequence dual to \eqref{ses}.

\subsection{Holomorphic orbifolds}

Let $V$ be a holomorphic VOA, i.e.\ a VOA whose only simple module is $V$ itself. Suppose a finite group $G$ is acting on $V$ by VOA automorphisms. Then the fixed point set $V^G$ is again a VOA, the orbifold VOA. 

It was argued in \cite{ki} that the category of modules of $V^G$ is equivalent to a (twisted) group-theoretic modular category $\Zc(G,\alpha)$ where $\alpha$ is a 3-cocycle on $G$. We assume for simplicity that $\alpha$ is trivial. Thus our modular category is $\Zc(G)$. This category can be described as the category of representations of the Drinfeld double $D(G)$, see \cite[Sec.\,IX.4.3,\,XIII.5]{Ka} or \cite[Sec.\,3.1]{da} for an explicit description of $\Zc(G)$. 

\medskip
Morita equivalence classes of simple special symmetric Frobenius algebras in $\Zc(G)$ were classified in \cite{os}. They are in one-to-one correspondence with conjugacy classes of pairs $(H,\gamma)$, where $H\subset G\times G$ is a subgroup and $\gamma\in H^2(H,\Cb^*)$ is a 2-cocycle. Simple commutative maximal special symmetric Frobenius algebras in $\Zc(G)\boxtimes \overline{\Zc(G)} = \Zc(G\times G)$ were described in \cite[Thm.\,3.5.1\,\&\,3.5.3]{da}. 
They are labelled by the same data (again making explicit the isomorphism \cite[Cor.\,3.25]{Kong:2008ci}).

The details of the following will again appear elsewhere.
The automorphism group $\Gamma(H,\gamma)$ of the simple commutative maximal algebra in $\Zc(G\times G)$ corresponding to the pair $(H,\gamma)$ is an extension
\be
  \hat H \to \Gamma(H,\gamma) \to St_{N_{G\times G}(H)/H}(\gamma) \ ,
\ee  
where $N_{G\times G}(H)$ is the normaliser of $H$ in $G\times G$.
The quotient $N_{G\times G}(H)/H$ has a well-defined action on the
cohomology $H^2(H,\Cb^*)$ by conjugation in each argument;
$St_{N_{G\times G}(H)/H}(\gamma)$ is the stabiliser of the class $\gamma$
with respect to this action.
In particular -- and in contrast with the previous example -- the automorphism group $\Gamma(H,\gamma)$  is often non-abelian. 

\section{Conclusion}\label{sec:conc}

We have shown that in a particularly well-understood class of quantum field theories, namely two-dimensional rational conformal field theories, all invertible duality transformations -- which are nothing but conformal isomorphisms -- can be implemented by one-dimensional domain walls (i.e.\ defect lines) provided both are compatible with the rational symmetry. In fact, given a rational VOA $V$ with category of representations $\Cc = \Rep V$, in Theorem~\ref{thm:main}
we proved an equivalence of groupoids between
\begin{itemize}
\item[-]
CFTs over $V \otimes \bar V$ and conformal isomorphisms acting as the identity on $V \otimes \bar V$ 
(the groupoid $\mathbf{A}(\CxC)$ in the algebraic formulation), and
\item[-]
CFTs over $V \otimes \bar V$ and (isomorphism classes of) invertible defect lines which preserve $V \otimes \bar V$
(the groupoid $\mathbf{P}(\Cc)$ in the algebraic formulation).
\end{itemize}

We would also like to note that this equivalence of groupoids has an application even for the best studied class of rational conformal field theories, the Virasoro minimal models \cite{Belavin:1984vu}. There, it is in principle possible to compute all bulk structure constants for all minimal models in the A-D-E classification of \cite{Cappelli:1987xt} using the methods of \cite{Fuchs:2001am,tft4}. But these are cumbersome to work with, and their conformal automorphisms have not been computed directly. Our result allows to instead compute fusion rules for bimodules, which is much easier to do (nonetheless they have not appeared in print explicitly for all minimal models). Our result also allows to make contact with \cite{Ruelle:1998zu}, where modular properties where used to investigate automorphisms of unitary minimal models.

\medskip

It turns out that our main result is not an isolated phenomenon. A result analogous to ours, but one categorical level higher, has recently been proved in \cite{enom,kitaev-kong}.  In \cite{enom}, a fully faithful embedding of 2-groupoids was obtained, where the role of $\mathbf{P}(\Cc)$ is taken by the 2-groupoid of fusion categories, bimodule categories, and isomorphism classes of equivalences of bimodule categories, and the role of $\mathbf{A}(\CxC)$ is taken by braided fusion categories, braided equivalences, and isomorphisms of braided equivalences. The functor is provided by the monoidal centre. This hints at a corresponding statement for Turaev-Viro theories. Although an axiomatic treatment of Turaev-Viro theories with domain walls is not yet available, a Hamiltonian version of Turaev-Viro theories -- the so-called Levin-Wen models \cite{Levin:2004mi} -- is carefully studied in \cite{kitaev-kong}. It is shown there that a bimodule category over two unitary tensor categories determines a domain wall between two bulk phases in a lattice model, and the monoidal centre describes anyon excitations in each bulk phase. Again, one has a one-to-one correspondence between invertible defects and equivalences (as braided tensor categories) between excitations in the bulk. 

\medskip

Even when staying within two-dimensional models, an important unanswered question is how much, if anything, of our analysis carries over from the maximally well-behaved class of models studied here to more complicated theories. For example, it would be very interesting (at least to us) to investigate logarithmic conformal field theories (see e.g.\ \cite{Gaberdiel:2001tr}) or topological conformal field theories \cite{Costello:2004ei}.

\bigskip\noindent
{\bf Acknowledgement}:
The authors would like to thank J\"urgen Fuchs for helpful comments on a draft of this paper.
AD thanks Max Planck Institut f\"ur Mathematik (Bonn) for hospitality and excellent working conditions. LK is supported in part by the Gordon and Betty Moore Foundation through Caltech's Center for the Physics of Information, and by the National Science Foundation under Grant No. PHY-0803371, and by the Basic Research Young Scholars Program of Tsinghua University.

\end{document}